\documentclass[a4paper]{article}       
\usepackage[a4paper]{geometry}
\usepackage{graphicx}
\usepackage{float}
\usepackage{cleveref}
\usepackage{graphicx,amssymb,amsmath}
\usepackage[noend]{algpseudocode}
\usepackage{algorithmicx,algorithm}
\usepackage{authblk}
\usepackage{amsthm}
\algrenewcommand\algorithmicrequire{\textbf{Input:}}
\algrenewcommand\algorithmicensure{\textbf{Output:}}
\begin{document}

\title{A Composable Coreset for $k$-Center in Doubling Metrics\thanks{The preliminary version of this paper has appeared in Proceedings of the 30th Canadian Conference on Computational Geometry, (CCCG 2018) \cite{conf}}
}
\date{}                     
\setcounter{Maxaffil}{0}
\renewcommand\Affilfont{\itshape\small}

\author[1]{Sepideh Aghamolaei}
\author[2]{Mohammad Ghodsi}

\affil[1]{Department of Computer Engineering,
        Sharif University of Technology, aghamolaei@ce.sharif.edu}
\affil[2]{Department of Computer Engineering,
        Sharif University of Technology, School of Computer Science, Institute for Research in Fundamental Sciences (IPM), ghodsi@sharif.edu}

\newtheorem{definition}{Definition}
\newtheorem{theorem}{Theorem}
\newtheorem{example}{Example}
\newtheorem{lemma}{Lemma}
\maketitle

\begin{abstract}
A set of points $P$ in a metric space and a constant integer $k$ are given. The $k$-center problem finds $k$ points as \textit{centers} among $P$, such that the maximum distance of any point of $P$ to their closest centers $(r)$ is minimized.

Doubling metrics are metric spaces in which for any $r$, a ball of radius $r$ can be covered using a constant number of balls of radius $r/2$. Fixed dimensional Euclidean spaces are doubling metrics. The lower bound on the approximation factor of $k$-center is $1.822$ in Euclidean spaces, however, $(1+\epsilon)$-approximation algorithms with exponential dependency on $\frac{1}{\epsilon}$ and $k$ exist.

For a given set of sets $P_1,\ldots,P_L$, a \textit{composable coreset} independently computes subsets $C_1\subset P_1, \ldots, C_L\subset P_L$, such that $\cup_{i=1}^L C_i$ contains an approximation of a measure of the set $\cup_{i=1}^L P_i$.

We introduce a $(1+\epsilon)$-approximation composable coreset for $k$-center, which in doubling metrics has size sublinear in $|P|$.
This results in a $(2+\epsilon)$-approximation algorithm for $k$-center in MapReduce with a constant number of rounds in doubling metrics for any $\epsilon>0$ and sublinear communications, which is based on parametric pruning.

We prove the exponential nature of the trade-off between the number of centers $(k)$ and the radius $(r)$, and give a composable coreset for a related problem called dual clustering. Also, we give a new version of the parametric pruning algorithm with $O(\frac{nk}{\epsilon})$ running time, $O(n)$ space and $2+\epsilon$ approximation factor for metric $k$-center.
\end{abstract}

\section{Introduction}
Coresets are subsets of points that approximate a measure of the point set.
A method of computing coresets on big data sets is composable coresets.
Composable coresets \cite{composable} provide a framework for adapting constant factor approximation algorithms to streaming and MapReduce models.
Composable coresets summarize distributed data so that the scalability is increased while keeping the desirable approximation factor and time complexity.

There is a general algorithm for solving problems using coresets which known by different names in different settings: mergeable summaries \cite{mergeable} and merging in a tree-like structure \cite{logarithmic} for streaming $(1+\epsilon)$-approximation algorithms, small space (divide and conquer) for constant factor approximations in streaming \cite{guha}, and composable coresets in MapReduce \cite{composable}. A consequence of using constant factor approximations instead of $(1+\epsilon)$-approximations with the same merging method is that it can add a $O(\log n)$ factor to the approximation factor of the algorithm on an input of size $n$.

Composable coresets \cite{composable} require only a single round and sublinear communications in the MapReduce model, and the partitioning is done arbitrarily.
\begin{definition}[Composable Coreset]
A composable coreset on a set of sets $\{S_i\}_{i=1}^L$ is a set of subsets
$C(S_i) \subset S_i$ whose union gives an approximation solution for
an objective function $f:(\cup_{i=1}^L S_i)\rightarrow \mathbf{R}$.
Formally, a composable coreset of a minimization problem is an $\alpha$-approximation if
\[
f(\cup_i S_i) \leq f(\cup_i C(S_i) ) \leq \alpha.f(\cup_i S_i),
\]
for a minimization problem. The maximization version is similarly defined.
\end{definition}
A \textit{partitioned composable coreset} is a composable coreset in which the initial sets are a partitioning, i.e. sets $\{S_i\}_{i=1}^L$ are disjoint.
Using Gonzalez's algorithm for $k$-center \cite{gonzalez}, Indyk, et al. designed a composable coreset for a similar problem known as the diversity maximization problem \cite{composable}.
Other variations of composable coresets are randomized composable coresets and mapping coresets.
Randomized composable coresets \cite{randomizedcomposable} share the same divide and conquer approach as other composable coresets and differ from composable coresets only in the way they partition the data.
More specifically, randomized composable coresets, randomly partitioning the input, as opposed to other composable coresets which make use of arbitrary partitioning.
Mapping coresets \cite{mapping} extend composable coresets by adding a mapping between coreset points and other points to their coresets and keep almost the same amount of data in all machines. Algorithms for clustering in $\ell^p$ norms using mapping coresets are known \cite{mapping}.
Further improvements of composable coresets for diversity
maximization \cite{composable} include lower bounds \cite{mine} and multi-round composable coresets in metrics with bounded doubling dimension~\cite{ediversity}.

\textit{Metric $k$-center} is an NP-hard problem for which $2$-approximation algorithms that match the lower bound for the approximation factor of this problem are known~\cite{vazirani,gonzalez}.
Among approximation algorithms for $k$-center is a parametric pruning algorithm, based on the minimum dominating set~\cite{vazirani}. In this algorithm, an approximate dominating set is computed on the disk graph of the input points. The running time of the algorithm is $O(n^3)$.
The greedy algorithm for $k$-center requires only $O(nk)$ time \cite{gonzalez} and unlike the algorithm based on the minimum dominating set\cite{vazirani}, uses $r$-nets \cite{rnet}.
A $(1+\epsilon)$-approximation coreset exists for $k$-center \cite{ptas} with size exponentially dependent on $\frac{1}{\epsilon}$.

Let the optimal radius of $k$-center for a point set $P$ be $r$. The problem of finding the smallest set of points that cover $P$ using radius $r$ is known as the \textit{dual clustering problem} \cite{charikar}.

Metric dual clustering (of $k$-center) has an unbounded approximation factor \cite{charikar}. In Euclidean metric, there exists a streaming $O(2^d d \log d)$-approximation algorithm for this problem~\cite{charikar}.
Also, any $\alpha$-approximation algorithm for the minimum disk/ball cover problem gives a $2$-approximation coreset of size $\alpha k$ for $k$-center, so $2$-approximation coresets of size $(1+\epsilon)k$ exist for this problem \cite{mdc}.
A greedy algorithm for dual clustering of $k$-center has also been used as a preprocessing step of density-based clustering~(DBSCAN)~\cite{dbscan}. Implementing DBSCAN efficiently in MapReduce is an important problem~\cite{mrdbscan,mrdbscan2,mrdbscan3,mrdbscan4,mrdbscan5}.

Randomized algorithms for metric $k$-center and $k$-median in MapReduce \cite{ene2011fast} exist. These algorithms take $\alpha$-approximation offline algorithms and return $(4\alpha+2)$-approximation and $(10\alpha+3)$-approximation algorithms for $k$-center and $k$-median in MapReduce, respectively. The round complexity of these algorithms depends on the probability of the algorithm for finding a good approximation.

Current best results on metric $k$-center in MapReduce have $2$ rounds and give the approximation factor $4$ \cite{kcenter1}. However, a $2$-approximation algorithm exists if the cost of the optimal solution is known~\cite{brief}. Experiments in \cite{kcenter2} suggest that running Gonzalez's algorithm on a random partitioning and an arbitrary partitioning results in the same approximation factor.

In doubling metrics, a $(2+\epsilon)$-approximation algorithm exists that is based on Gonzalez's greedy algorithm \cite{other}. The version with outliers has also been discussed \cite{other,ding}.

\subsection*{Warm-Up}
Increasing the size of coresets in the first step of computing composable coresets can improve the approximation factor of some problems.
The approximation factor of $k$-median algorithm of \cite{guha} is $2c(1+2b)+2b$, where $b$ and $c$ are the approximation factors of $k$-median and weighted $k$-median, respectively. This algorithm computes a composable coreset, where a coreset for $k$-median is the set of $k$ medians weighted by the number of points assigned to each median.

A pseudo-approximation for $k$-median finds $k+O(1)$ median and has approximation factor $1+\sqrt{3}+\epsilon$ \cite{pseudo}.
Using a pseudo-approximation algorithm in place of $k$-median algorithms in the first step of \cite{guha}, it is possible to achieve a better approximation factor for $k$-median using the same proof as \cite{guha}. Since any pseudo-approximation has a cost less than or equal to the optimal solution; replacing them will not increase the cost of clustering.

The approximation factor using \cite{weighted} as weighted $k$-median coresets is $91.66$, while the best $k$-median algorithm would give a $99.33$ factor using the same algorithm ($b=1+\sqrt{3}$). The lower bound on the approximation factor of this algorithm using the same weighted $k$-median algorithm but without pseudo-approximation is $63.09$ ($b=1+\frac{2}{e}$).

\subsection*{Contributions}
We give a $(1+\epsilon)$-approximation coreset of size $(\frac{4}{\epsilon})^{1+2b}k$ for $k$-center in metric spaces with doubling dimension $b$.
Using composable coresets, our algorithm generalizes to MapReduce setting, where it becomes a $(1+\epsilon)$-approximation coreset of size $(\frac{4}{\epsilon})^{1+2b}\frac{n}{m}k$, given memory $m$, which is sublinear in the input size $n$.

\begin{table}[h]
\centering
\begin{tabular}{|l|c|l|}
\hline
Conditions & Approx. & Reference\\
\hline
Metric $k$-center: &&\\
$O(1)$-rounds & $4$ & \cite{kcenter1} (greedy), \Cref{theorem:gkcenter} (parametric pruning)\\
$O(\log_{1+\epsilon}^{\Delta})$ rounds & $2+\epsilon$ & \cite{brief} (parametric pruning)\\
Lower bound & $2$ & offline \cite{vazirani}\\
\hline
Doubling metrics: &&\\
$O(1)$-rounds & $2+\epsilon$ & \cite{other} (greedy), \Cref{theorem:kcenter}~(parametric pruning)\\
Lower bound & $1.822$ & \cite{lowerbound}\\
\hline
\hline
Dual clustering: &&\\
General metrics & $O(\log n)$ & min dominating set \cite{vazirani}, composable coreset \cite{composable}\\
Doubling metrics & $O(1)$ & \Cref{theorem:dual}\\
\hline
\end{tabular}
\caption{Summary of results for $k$-center and dual clustering in MapReduce. $\Delta$ is the diameter of the point-set.}
\label{table:kcenter}
\end{table}

Using the composable coreset for dual clustering, we find a $(2+\epsilon)$-approximation composable coreset for $k$-center, which has a sublinear size in metric spaces with constant doubling dimension. More specifically, if an $\alpha$-approximation exists for doubling metrics, our algorithm provides $(\alpha+\epsilon)$-approximation factor.
It empirically improves previous metric $k$-center algorithms \cite{kcenter1,kcenter2} in MapReduce.
A summary of results on $k$-center is shown in \Cref{table:kcenter}. Note that for the MapReduce model, each round can take a polynomial amount of time, however, the space available to each machine is sublinear.

Our algorithm achieves a trade-off between the approximation factor and the size of coreset (see \cref{fig:plot}). The approximation factor of our algorithm and the size of the resulting composable coreset for $L$ input sets are $\alpha=2+\epsilon$ and $kL\beta$, respectively. This trade-off is the main idea of our algorithm.

\begin{figure}[h]
\centering
\includegraphics[scale=0.8]{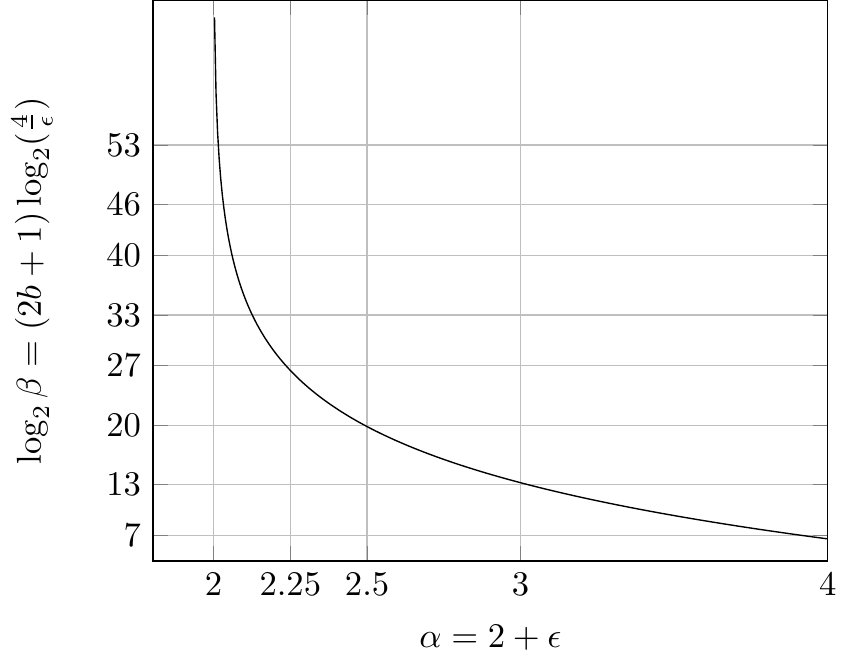}
\caption{Space-approximation factor trade-off of our $\alpha$-approx. coreset of size $\beta kL$ for $k$-center in Euclidean plane.}
\label{fig:plot}
\end{figure}

Our composable coresets give single-pass streaming algorithms and $1$-round approximation algorithms in MapReduce with sublinear communication, since each coreset is communicated once, and the size of the coreset is constant.

\section{Preliminaries}
First, we review some basic definitions, models and algorithms in computational geometry and MapReduce.

\subsection{Definitions}
Some geometric definitions and notations are reviewed here, which have been used in the rest of the paper.
\begin{definition}[Metric Space]
A (possibly infinite) set of points $P$ and a distance function $d(.,.)$ create a \textit{metric space} if the following three conditions hold:
\begin{itemize}
\item $\forall p,q\in P \quad d(p,q)=0 \Leftrightarrow p=q$
\item $\forall p,q\in P \quad d(p,q) = d(q,p)$
\item $\forall p,q,t \in P \quad d(p,q)+d(q,t) \geq d(p,t)$, known as triangle inequality
\end{itemize}
\end{definition}

Metrics with bounded doubling dimension are called \textit{doubling metrics}. Constant dimension Euclidean spaces under $\ell^p$ norms and Manhattan distance are examples of doubling metrics.

Doubling constant \cite{doublingdef} of a metric space is the number of balls of radius $r$ that lie inside a ball of radius $2r$. The logarithm of doubling constant in base $2$ is called \textit{doubling dimension}. Many algorithms have better approximation factors in doubling metrics compared to general metric spaces. The doubling dimension of Euclidean plane is $\log_2 7$.
\begin{definition}[Doubling Dimension \cite{doublingdef}]\label{def:doubling}
For any point $x$ in a metric space and any $r\geq 0$, if the ball of radius $2r$
centered at $x$ can be covered with at most $2^b$ balls of radius $r$, we say
the doubling dimension of the metric space is $b$.
\end{definition}

$k$-Center is an NP-hard clustering problem with clusters in shapes of $d$-dimensional balls.
\begin{definition}[Metric $k$-Center \cite{vazirani}]
Given a set $P$ of points in a metric space, find a subset of $k$ points as cluster centers $C$ such that
\[
\forall p\in P, \min_{c\in C} d(p,c) \leq r
\]
and $r$ is minimized.
\end{definition}
The best possible approximation factor of metric $k$-center is $2$ \cite{vazirani}.

Geometric intersection graphs represent intersections between a set of shapes. For a set of disks, their intersection graph is called a disk graph.
\begin{definition}[Disk Graph]
For a set of points $P$ in a metric space with distance function $d(.,.)$ and a radius $r$, the disk graph of $P$ is a graph whose vertices are $P$, and whose edges connect points with distance at most $2r$.
\end{definition}

\begin{definition}[Dominating Set]
Given a graph $G=(V,E)$, the smallest subset $Q\subset V$ is a minimum dominating set, if $\forall v\in V, v\in Q \vee \exists u\in Q : (v,u)\in E$.
\end{definition}

We define the following problem as a generalization of the dual clustering of \cite{charikar} by removing the following two conditions: the radius of balls is $1$, and the set of points are in $\mathbf{R}^d$.
\begin{definition}[Dual Clustering]
Given a set of points $P$ and a radius $r$, the dual clustering problem finds the smallest subset of points as centers $(C), C\subset P$ such that the distance from each point to its closest center is at most $r$.
\end{definition}
\subsection{An Approximation Algorithm for Metric $k$-Center}
Here, we review the parametric pruning algorithm of~\cite{vazirani} for metric $k$-center.

\begin{algorithm}[H]
\caption{Parametric Pruning for $k$-Center \cite{vazirani}}
\label{alg:parametric}
\begin{algorithmic}
\Require{A metric graph $G=(V,E)$, an integer $k$}
\Ensure{A subset $C\subset V, |C|\leq k$}
\State{Sort $E$ such that $e_1 \leq e_2 \leq \cdots \leq e_{|E|}$.}
\State{$G'=(V,E')\gets (V,\emptyset)$}
\For{$i=1,\ldots,|E|$}
	\State{$E'\gets E'\cup \{e_i\}$}
	\State{Run \cref{alg:cbc} on $G'$.}
	\If{$|IS|\leq k$}
		\Return {$IS$}
	\EndIf
\EndFor
\end{algorithmic}
\end{algorithm}

Using this algorithm on a metric graph $G$, a $2$-approximation for the optimal radius $r$ can be determined.
In \cref{alg:parametric}, edges are added by increasing order of their length until reaching $r$.
Given this radius, another graph $(G')$ is built, where edges exist between points within distance at most $r$ of each other.

\begin{algorithm}[H]
\caption{Approximate dominating set of $G$ \cite{vazirani}}
\label{alg:cbc}
\begin{algorithmic}
\Require{A metric graph $G'=(V,E)$}
\Ensure{A subset $C\subset V$}
\State{$G'^2 \gets G'$}
	\For{$\forall (u,t),(t,v) \in E$}
	\State{Add $(u,v)$ to $G'^2$.}
	\EndFor
	\State{Find a maximal independent set $IS$ of $G'^2$}\\
\Return $IS$
\end{algorithmic}
\end{algorithm}

Hence, by definition, a minimum dominating set of $G'$ is an optimal $k$-center of $G$.
Every cluster is a star in $G'$ which turns into a clique in $G'^2$. Therefore, a maximal independent set of $G'^2$ chooses at most one point from each cluster. \Cref{alg:cbc} computes $G'^2$ and returns a maximal independent set of $G'^2$.

Computing a maximal independent set takes $O(|E|)$ time. The graph $G'^2$ in \Cref{alg:cbc} only changes in each iteration of \Cref{alg:parametric} around the newly added edge, so, updating the previous graph and $IS$ takes $O(n)$ time. Therefore, the time complexity of \Cref{alg:parametric} is $O(|E|\cdot n)=O(n^3)$.
\section{A Coreset for Dual Clustering in Doubling Metrics}
In this section, we prove a better approximation offline coreset for the dual clustering problem. Our method is based on \Cref{alg:parametric} which first builds the disk graph with radius $r$, then covers this graph using a set of stars. We prove the maximum degree of those stars is $D^2$, where $D$ is the doubling constant. The result is an approximation algorithm for dual clustering in doubling metrics.
\subsection{Algorithm}
We add a preprocessing step to \Cref{alg:parametric} to find a better approximation factor for $k$-center and dual clustering problems.
\begin{algorithm}[H]
\caption{A Coreset for $k$-Center}
\label{alg:offline}
\begin{algorithmic}
\Require{A set of points $P$, an integer $k$ or a radius $r$}
\Ensure{A subset $C\subset P, |C|\leq k$}
\If{$k$ is given in the input}
\State{Compute a $2$-approximation solution for $k$-center (radius $r$).}
\EndIf
\State{$E\gets$ all pairs of points with distance at most $r/2$.}
\State{Run \cref{alg:cbc} on $G=(P,E)$ to compute $IS$.}\\
\Return {$IS$}
\end{algorithmic}
\end{algorithm}

\subsection{Analysis}
Unlike in general metric spaces, $k$-center in doubling metrics admits a space-approximation factor trade-off.
More specifically, doubling or halving the radius of $k$-center changes the number of points in the coreset by a constant factor since the degrees of vertices in the minimum dominating set are bounded in those metric spaces.

\begin{lemma}\label{lemma:kissing0}
For each cluster $C_i$ of \Cref{alg:offline} with radius $r'$, the maximum number of points $(\Delta+1)$ from $C_i$ that are required to cover all points inside $C_i$ with radius $r'/2$ is at most $D^2$, i.e.
\[
(\Delta+1) \leq D^2,
\]
where $D$ is the doubling constant of the metric space.
\end{lemma}
\begin{proof}
Assume a point $p\in IS$ returned by \Cref{alg:offline}.
By the definition of doubling metrics, there are $D$ balls of radius $r'/2$ centered at $b_1,\ldots,b_{D}$ called $B_1,\ldots,B_D$ that cover the ball of radius $r'$ centered at $p$, called $B$.
\[
\forall q\in B, \exists B_i, i=1,\ldots,D : d(p,b_i)\leq r'/2
\]
Repeating this process for each ball $B_i$ results in a set of at most $D$ balls $(B'_{i,1},\ldots,B'_{i,D})$ of radius $r'/4$ centered at $b'_{i,1},\ldots,b'_{i,D}$.
\[
\forall q\in B'_{i,j}, d(b'_{i,j},q)\leq r'/4
\]
Choose a point $p_{i,j}\in P\cap B'_{i,j}$. Using triangle inequality,
\begin{align*}
\forall q\in B'_{i,j}, d(p_{i,j},q) &\leq d(p_{i,j},b'_{i,j})+d(b'_{i,j},q)\\
&\leq r'/4+r'/4=r'/2.\*
\end{align*}

We claim any minimal solution needs at most one point from each ball $B'_{i,j}$. By contradiction, assume there are two point $p_{i,j},q'$ in the minimal solution that lie inside a ball $B'_{i,j}$. After removing $q'$, the ball with radius $r'/2$ centered at $p_{i,j}$ still covers $B'_{i,j}$, since:
\begin{align*}
\forall q\in P, \exists B_i ,B'_{i,j}\ni q,p_{i,j}&\\
d(q,p_{i,j}) &\leq d(q,b'_{i,j})+d(b'_{i,j},p_{i,j})\\
&\leq r'/4+r'/4=r/2'.\*
\end{align*}
Then we have found a point $(q')$ whose removal decreases the size of the solution, which means the solution was not minimal. So, the size of any minimal set of points covering $B$ is at most $D^2$.
\end{proof}

\begin{figure}[h]
\centering
\includegraphics[scale=0.7]{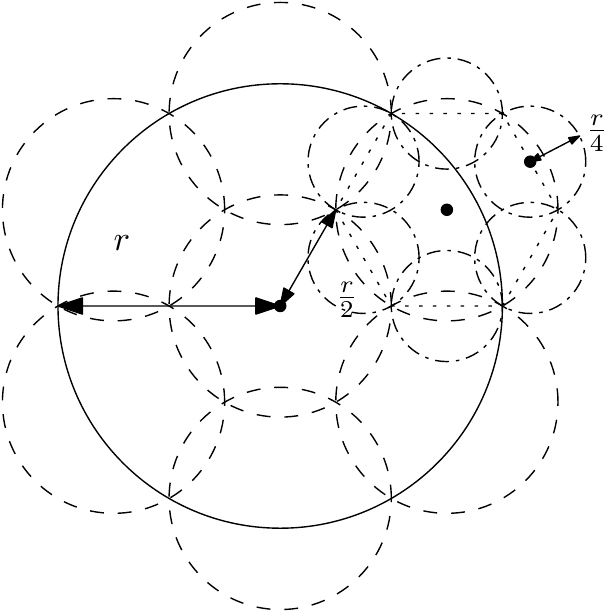}
\caption{Applying the doubling dimension bound twice (Lemma~\ref{lemma:kissing0}).}
\label{fig:kissing}
\end{figure}

\begin{lemma}\label{lemma:kissing}
In a metric space with doubling constant $D$, if a dual clustering with radius $r$ has $k$ points, then a dual clustering with radius $r/2$ exists which has $D^2 k$ points.
\end{lemma}
\begin{proof}
Let $p$ be a center in the $k$-center problem. Based on the proof of Lemma~\ref{lemma:kissing0}, there are $\Delta$ vertices adjacent to $p$ that cover the points inside the ball of radius $r$ centered at $p$, using balls of radius $r/2$ and a ball of radius $r/2$ centered at $p$. By choosing all these vertices as centers, it is possible to cover all input points $P$ with radius $r/2$. Using the same reasoning for all clusters, it is possible to cover all points using $(\Delta+1) k$ centers. Using the bound in Lemma~\ref{lemma:kissing0}, these are $D^2 k$ centers.
\end{proof}

\begin{theorem}\label{theorem:dual}
The approximation factor of \Cref{alg:offline} is $D^2$ for the dual clustering.
\end{theorem}
\begin{proof}
Since the radius of balls in Lemma~\ref{lemma:kissing} is at most the optimal radius for $k$-center, the approximation factor of dual clustering is the number of points chosen as centers divided by $k$, which is $D^2$.
\end{proof}

\begin{theorem}\label{theorem:krcenter}
The approximation factor of the coreset for $k$-center in \Cref{alg:offline} is $2^{-R}$ and its size is $D^{2(R+1)}k$.
\end{theorem}
\begin{proof}
Applying Lemma~\ref{lemma:kissing} halves the radius and multiplies the number of points by $D^2$. So, applying this lemma $R$ times gives $(D^2)^{R+1}k$ points since it might be the case that in the first step of the algorithm the optimal radius was found, and we divided it by $2$. The radius remains $\frac{r}{2^R}$ because of the case where we had found a $2$-approximation.
\end{proof}

\begin{theorem}\label{theorem:kcenter0}
\Cref{alg:offline} given $(\frac{4}{\epsilon})^{2\log_2 D}k$ as input, is a $(1+\epsilon)$-approximation coreset of size $(\frac{4}{\epsilon})^{2\log_2 D}k$ for the $k$-center problem.
\end{theorem}
\begin{proof}
For $R=\lceil \log_2 \frac{2}{\epsilon} \rceil$, the proof of \Cref{theorem:krcenter} gives $(\frac{4}{\epsilon})^{2\log_2 D}$ points and radius $r\epsilon$.
Assume $O$ is the set of $k$ centers returned by the optimal algorithm for point-set $P$, and $C$ is the set of centers returned by running the optimal algorithm on the coreset of $P$. For any point $p\in P$, let $o$ be the center that covers $p$ and $c$ be the point that represents $o$ in the coreset.
Using triangle inequality:
\[
d(p,c) \leq d(p,o)+d(o,c) \leq r+r\epsilon=(1+\epsilon)r
\]
So, computing a $k$-center on this coreset gives a $(1+\epsilon)$-approximation.
\end{proof}
\section{A Composable Core-Set for $k$-Center in Doubling Metrics}
Our general algorithm for constructing coresets based on dual clustering has the following steps:
\begin{itemize}
\item Compute the cost of an approximate solution $(X)$.
\item Find a composable coreset for dual clustering with cost $X$.
\item Compute a clustering on the coreset.
\end{itemize}
In this section, we use this general algorithm for solving $k$-center.
\subsection{Algorithm}
Knowing the exact or approximate value of $r$, we can find a single-round $(2+\epsilon)$-approximation for metric $k$-center in MapReduce.
Although the algorithm achieves the aforementioned approximation factor, the size of the coreset and the communication complexity of the algorithm depend highly on the doubling dimension.

\begin{algorithm}[h]
\caption{$k$-Center}
\label{alg:kcenter}
\begin{algorithmic}[1]
\Require{A set of sets of points $\cup_{i=1}^L S_i$, a $k$-center algorithm}
\Ensure{A set of $k$ centers}
\State{Run a $k$-center algorithm on each $S_i$ to find the radius $r_i$.}
\State{Run \Cref{alg:cbc} on the disk graph of each set $S_i$ with radius $\frac{\epsilon r_i}{2}$ locally to find $C(S_i)$.}
\State{Send $C(S_i)$ to set $1$ to find the union $\cup_i C(S_i)$.}
\State{Run a $2$-approximation $k$-center algorithm on $\cup_{i=1}^L C(S_i)$ to find the set of centers $C$.}\\
\Return{$C$.}
\end{algorithmic}
\end{algorithm}

Based on the running time of \Cref{alg:cbc} and Gonzalez's algorithm, the running time of \Cref{alg:kcenter} is $\sum_i [O(k\cdot |S_i|)+O(|S_i|^2)]+O(k\sum_i |C(S_i)|)=O(kn)$. Since the sum of running times of machines is of the same order as the best sequential algorithm, \Cref{alg:kcenter} is a work-efficient parallel algorithm.

\begin{figure}[h]
\centering
\includegraphics[scale=0.8]{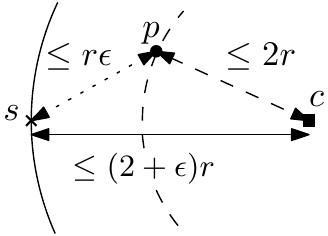}
\caption{The dominating set on $\cup_i C(S_i)$ covers $\cup_i S_i$ with radius $(2+\epsilon)$(Theorem \ref{theorem:kcenter}).}
\label{fig:kcenter}
\end{figure}

We review the following well-known lemma:
\begin{lemma}\label{lemma:kcenter}
For a subset $S\subset P$, the optimal radius of the $k$-center of $S$ is at most twice the radius of the $k$-center of $P$. 
\end{lemma}
\begin{proof}
Consider the set of clusters $O_i$ in the optimal $k$-center of $P$ centered at $c_i,i=1,\ldots,k$ with radius $r$. If $c_i\in S$, then the points of $O_i\cap S$ are covered by $c_i$ with radius $r$, as before. Otherwise, select an arbitrary point in $O_i\cap S$ as the new center $c'_i$. Using the triangle inequality on $c_i,c'_i$ and any point $p\in O_i\cap S$:
\[
d(p,c'_i) \leq d(p,c_i)+d(c_i,c'_i) \leq r+r=2r
\]
Since $c'_i$ was covered using $c_i$ with radius $r$. So, the set $S\cap O_i$ can be covered with radius $2r$. Note that since we choose at most one point from each set, the number of new centers is at most $k$.
\end{proof}

\begin{theorem}\label{theorem:kcenter}
The approximation factor of \Cref{alg:kcenter} is $2+\epsilon$ for metric $k$-center.
\end{theorem}
\begin{proof}
Let $r$ be the optimal radius of $k$-center for $\cup_i S_i$.
Since $\cup_i C(S_i) \subset \cup_i S_i$, using Lemma~\ref{lemma:kcenter}, the radius of $k$-center for $\cup_i C(S_i)$ is at most $2r$. The radius of $k$-center inside each set $S_i$ is at most $2r$ for the same reason.
The algorithm computes a covering $S_i$ with balls of radius $r_i \epsilon /2$.
Based on the fact that offline $k$-center has $2$-approximation algorithms and the triangle inequality, the approximation factor of the algorithm proves to be $(2+\epsilon)$-approximation (\Cref{fig:kcenter}). Let $p=\arg \min_{p\in \cup_i C(S_i)} dist(s,p)$, then
\begin{align*}
\forall s \in S_i \exists c \in C, d(s,c) &\leq d(s,p)+d(p,c) \leq r'+r_i\epsilon/2 \\
&\leq 2r+2r\epsilon/2 =(2+\epsilon)r\*
\end{align*}
where $r'$ is the radius of the offline $k$-center algorithm on $C$.
\end{proof}

\subsection{Analysis}

\begin{lemma}\label{lemma:kissing2}
In a metric space with doubling constant $D$, the union of dual clusterings of radius $r$ computed on sets $S_1,\ldots,S_L$ is a $(L\times D^{2\log_2 \frac{8}{\epsilon}})$-approximation for the dual clustering of radius $r(1+\epsilon)$ of their union $(\cup_{i=1}^L S_i)$.
\end{lemma}
\begin{proof}
Each center in the dual clustering with radius $r$ of $P=(\cup_{i=1}^L S_i)$ has at most $\Delta$ adjacent vertices covered by this center.
Consider a point $p\in P$ covered by center $c$ in a solution for $P$. If $p$ and $c$ belong to the same set $S_i$, assign $p$ to $c$. Otherwise, pick any point that was previously covered by $c$ as the center that covers $p$.

While this might increase the radius by a factor $2$, it does not increase the number of centers in each set. Since the algorithm uses radius $\epsilon .r /2$, it increases the number of centers to $D^{2\log_2 \frac{8}{\epsilon}}k$ (based on \Cref{theorem:krcenter} for $R=\frac{4r}{\epsilon r/2}$) but keeps the approximation factor of the radius to $1+\epsilon$.
There are $L$ such sets, so the size of the coreset is $L\times D^{2\log_2 \frac{8}{\epsilon}}k$.
\end{proof}

\begin{theorem}\label{theorem:mrdbscan}
\Cref{alg:kcenter} returns a coreset of size $O(kL)$ for $k$-center in metric spaces with fixed doubling dimension.
\end{theorem}
\begin{proof}
The coreset of each set $S_i$ has a radius $r_i$ varying from the optimal radius $(r=r_i)$ to $2\beta.r$, where $\beta$ is the approximation factor of the offline algorithm for $k$-center.
Clearly, the lower bound holds because any radius is at least as much as the optimal (minimum) radius, which means $r\leq r_i$; and Lemma~\ref{lemma:kcenter} when applied to $S_i \subset \cup_i S_i$, yields the upper bound.
\[
r\leq r_i \leq 2\beta.r \Rightarrow \frac{r\epsilon}{4\beta} \leq \frac{r_i\epsilon}{4\beta} \leq \frac{\epsilon r}{2}
\]
Reaching value $r\epsilon$ requires applying \Cref{theorem:kcenter} at most $\log_2 \frac{4\beta}{\epsilon}$ times.

The size of the resulting coreset is therefore at most
\[
(4^{\log_2 D})^{\log_{2} \frac{4\beta}{\epsilon}}kL = (\frac{4\beta}{\epsilon})^{2(\log_2 D)}kL.
\]
Here, we use the best approximation factor for metric $k$-center $(\beta=2)$, which gives a coreset of size $(\frac{8}{\epsilon})^{2(\log_2 D)}kL=O(kL)$ for fixed $\epsilon$.
\end{proof}
\subsection{Generalized Approximation Factor}
We prove that any $2$-approximation algorithm that does not choose a center from the points of another center can be used instead of Gonzalez's algorithm in the MapReduce algorithm of \cite{kcenter1}, and a similar proof will give the approximation factor $4$.
\Cref{alg:gkcenter} shows the generalized algorithm.

\begin{algorithm}[h]
\caption{Generalized Metric $k$-Center in MapReduce}
\label{alg:gkcenter}
\begin{algorithmic}[1]
\Require{A set of sets of points $\cup_{i=1}^L S_i$}
\Ensure{A set of $k$ centers}
\State{Run a $k$-center algorithm on each $S_i$ to find the radius $r_i$ and the set of centers $C(S_i)$.}
\State{Send $C(S_i)$ to set $1$ to find the union $\cup_i C(S_i)$.}
\State{Run a $2$-approximation $k$-center algorithm on $\cup_{i=1}^L C(S_i)$ to find the set of centers $C$.}\\
\Return{$C$.}
\end{algorithmic}
\end{algorithm}

\begin{theorem}\label{theorem:gkcenter}
\Cref{alg:gkcenter} given an $\alpha$-approximation metric $k$-center algorithm with $\alpha \geq 2$ which does not choose a center from the points of another cluster, finds a $2\alpha$-approximation solution.
\end{theorem}
\begin{proof}
Assume $r^*$ is the optimal $k$-center radius of $\cup_{i=1}^L S_i$. We prove that $C(S_i)$ covers $S_i$ with radius at most $\alpha \cdot r^*$. Suppose there is a point $p\in S_i$ whose distance to its nearest point from $C(S_i)$ is more than $\alpha r^*$, so $r_i\geq \alpha r^*$. The distance between each pair of points from $C(S_i)$ is at least $r_i$, since the algorithm never chooses a point as a center if it is within distance $r_i$ of another center. Therefore, the set $\{p\} \cup C(S_i)$ has $k+1$ points with distance at least $r_i$ from each other. There are at most $k$ optimal clusters, so at least two of these points must lie inside a cluster, which means their distance is at most $2r^*$. This means that $r_i \leq 2 r^*$, which contradicts the previous bound $r_i \geq \alpha r^*$.

A similar proof follows for $\cup_{i=1}^L C(S_i)$ and $C$. Using triangle inequality, the distance from any point $p$ to its local center $c(p)$ and its final center $c'(c(p))$ is bounded by:
\[
d(p,c(p))+d(c(p),c'(c(p))) \leq \alpha r^*+\alpha r^*=2\alpha r^*.
\]
\end{proof}
Note that the parametric pruning algorithm finds a dominating set by computing a maximal independent set, so the centers returned by this algorithm do not lie inside each others' clusters.
\subsection{A $(1+\epsilon)$-Composable Core-Set}
The composable coreset for $k$-center in doubling metrics can be used to obtain a $(1+\epsilon)$-approximation for constant $\epsilon$ and $k$. All these results also hold for dual clustering, as a result of the proven trade-off between $r$ and $k$.

\begin{algorithm}[h]
\caption{$k$-Center for Fixed $k$}
\label{alg:ikcenter}
\begin{algorithmic}[1]
\Require{A set of sets of points $\cup_{i=1}^L S_i$}
\Ensure{A set of $k$ centers}
\State{Run a $k$-center algorithm locally to find $r_i$.}
\State{Run \Cref{alg:cbc} on the disk graph of each set $S_i$ with radius $\frac{\epsilon r_i}{2}$ locally to find $C(S_i)$.}
\State{Send $C(S_i)$ to set $1$ to find the union $\cup_i C(S_i)$.}
\State{Run \Cref{alg:offline} on $\cup_i C(S_i)$, and let $C'$ be its output.}
\State{Run the optimal $k$-center of $C'$ by checking all $\binom{|C'|}{k}$ possible subsets, and let $T$ be its output.}\\
\Return{$T$}
\end{algorithmic}
\end{algorithm}

\begin{theorem}
\Cref{alg:ikcenter} gives a $(1+\epsilon)$-approximation for $k$-center in doubling metrics, for fixed $k$ and $\epsilon$.
\end{theorem}
\begin{proof}
The approximation factor of $C$ is $1+\epsilon$ and its size is $(\frac{8}{\epsilon})^{2\log_2 D}kL$, based on \Cref{theorem:mrdbscan}.
Repeating the core-set computation gives the approximation factor $(1+\epsilon)^2 \approx 1+2\epsilon$, and has size $(\frac{4}{\epsilon})^{2\log_2 D}k$ as proved in \Cref{theorem:kcenter0}.
Checking all possible choices for $k$ centers from $C'$ takes polynomial time, for fixed $k$ and $\epsilon$, since:
\[
\binom{(\frac{4}{\epsilon})^{2\log_2 D}k}{k}\leq (\frac{e(\frac{4}{\epsilon})^{2\log_2 D}k}{k})^k=(\frac{4}{\epsilon})^{2k\log_2 D}.
\]
Since the last step was optimal, the approximation factor of $T$ for $k$-center is $1+\epsilon$.
\end{proof}
\section{The Exponential Nature of The Trade-off Between $r$ and $k$}
The same constructive algorithm yields an exponential lower bound on the trade-off between $r$ and $k$ of $k$-center.

We build the following example by placing a point at the center of each ball from ball covering problem using balls of radius $r/2$ to cover a ball of radius $r$, and repeating this process recursively.
\begin{example}\label{example}
Cover the ball of radius $r$ with $D$ balls of radius $r/2$, where $D$ is the doubling constant of the metric space. Repeat this process with each of the balls of radius $r/2$.
The number of balls in the $t$-th iteration of this process is $D^t$ and their radius is $\frac{r}{2^t}$.
\end{example}

\begin{lemma}\label{lemma:circlepacking}
A circle packing of radius $R$ with circles of radius $r/2$ is an upper bound for the ball cover using circles of radius $r$, and the circle packing using circles of radius $r$ is a lower bound for the ball cover of radius $r$.
\end{lemma}
\begin{proof}
The circle packing has the maximum number of circles, so there is no room for more circles in the empty spaces between those circles. Therefore, increasing the radius of circles to twice the previous radius will cover the circle of radius $R$.
So, the circle packing for circles of radius $r$ gives an upper bound on the minimum number of circles required to cover a circle of radius $R$.

On the other hand, circle packing for circles of radius $r/2$ is a lower bound for the minimum number of circles required to cover the circle of radius $R$, since all those circles are disjoint.
\end{proof}

\begin{theorem}\label{theorem:tradeoff}
The optimal trade-off between $k$ and $r$ is exponential.
\end{theorem}
\begin{proof}
Based on \Cref{lemma:circlepacking}, \Cref{theorem:krcenter} gives both a lower bound and an upper bound on the trade-off between $k$ and $r$ for $O(k)$ points and radius $r/2$, within a constant factor for doubling metrics.
\Cref{lemma:kissing} gives the upper bound $D^2$ for each step, and the lower bound in \Cref{example} is $D$.
Substituting this bound in the trade-off of \Cref{theorem:krcenter} gives the ratio between the upper bound and the lower bound of $k$ in this trade-off which is $(\frac{4}{\epsilon})^{\log_2 D}k$,
where $\epsilon$ is the radius of the balls used for covering the points.
\end{proof}
Better trade-offs in $\mathbb{R}^2$ and $\mathbb{R}^3$ can be achieved by replacing $D^2$ with the square of the bound from circle/sphere covering for radius $\frac{r}{\epsilon}$ in \Cref{theorem:tradeoff} instead.
\section{A Comparison of The Algorithms for Metric $k$-Center}\label{sec:compare}
We consider variations of Gonzalez's greedy algorithm and the parametric pruning algorithm in which arbitrary choices are replaced by random ones. In the worst case, even the randomized version of these algorithms cannot achieve an approximation factor better than $2$.

We also prove the solutions of Gonzalez's algorithm are a subset of the solutions of the parametric pruning algorithm.

\begin{lemma}\label{lemma:rg}
There are instances in which randomized Gonzalez's algorithm cannot do better than $2$-approximation in the best case.
\end{lemma}
\begin{proof}
We prove this lemma by the counterexample of \Cref{fig:gonzalez}.
The measures of the example are as follows:
\[
d(P_i,P_{i+1})=\frac{1}{2},\quad i=1,\cdots,4
\]
The farthest neighbor computation prevents the algorithm from choosing the optimal solution $\{P_2,P_4\}$ with cost $\frac{1}{2}$, since it chooses at least one of the endpoints $P_1,P_5$. Therefore, the cost of the solution computed by randomized Gonzalez's algorithm is $1$. So, the approximation factor is $\frac{1}{1/2}=2$.
\begin{figure}[h]
\centering
\includegraphics[scale=0.8]{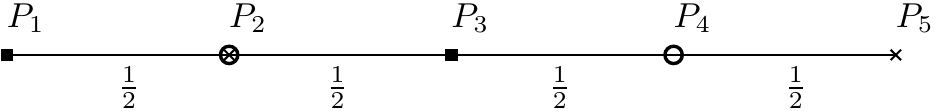}
\caption{Gonzalez's algorithm never finds $P'_1,P'_2,P'_3$ as the solution.}
\label{fig:gonzalez}
\end{figure}
\end{proof}

\begin{lemma}\label{lemma:gp}
The solutions of Gonzalez's algorithm are a subset of the solutions of the parametric pruning algorithm.
\end{lemma}
\begin{proof}
Let $r$ be the radius and $C$ be the set of centers computed by Gonzalez's algorithm, after removing the last centers if they do not decrease the cluster radius.
Consider the graph $G=(P,E)$ such that $P$ is the set of input points and $E$ is the set of all pairs of points with distance at most $r$.
By the anti-cover property of Gonzalez's algorithm, $C$ is an independent set of $G^2$.

Since the maximal independent set algorithm visits vertices in an arbitrary order, use the order of visiting used in Gonzalez's algorithm.
Consider an instance of the parametric pruning algorithm that at the $i$-th step, visits the points of $P$ in the order of Gonzalez's algorithm after it has chosen $i$ centers.
In such an instance, all the edges between the points of $C$ and their corresponding members from $P$ have already appeared in the sorted list of edges, since they have a lower edge weight than $r$. Also, there are no edges between the points of $C$, since Gonzalez's algorithm chooses the farthest point from previous points, so the distance between the centers is more than $r$.
Therefore, $C$ is a maximal independent set of the disk graph of radius $r$.
All the radii less than $r$ that are checked by the parametric pruning algorithm will fail since $r$ is the minimum radius that covers $P$ using points of $C$.
For radius $r$, the parametric pruning algorithm finds the solution with $C$ as centers.

We proved that there is an execution of the maximal independent set algorithm on the square of $G$ that finds $C$ as the set of centers.
\end{proof}

\begin{lemma}\label{lemma:po}
There are examples in which randomized parametric pruning algorithm for $k$-center cannot do better than a $2$-approximation in the best case.
\end{lemma}
\begin{proof}
Any solution in the form of a dominating set of the unit disk graph that is not also an independent set is a solution that the parametric pruning algorithm cannot find.
See \Cref{fig:parametric} for an example. In this example, $\{p_1,p_2\}$ are an optimal solution, but the parametric pruning algorithm does not find them. $\{p_1,q\}$ is a solution that the parametric pruning algorithm can find, because it is an independent set.
\begin{figure}[h]
\centering
\includegraphics{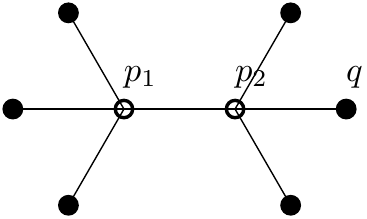}
\caption{$\{p_1,p_2\}$ is a dominating set which is not also an independent set.}
\label{fig:parametric}
\end{figure}
\end{proof}
\Cref{lemma:rg,lemma:gp,lemma:po} show the relation between the $2$-approximation solutions of $k$-center, as stated in the following theorem.
\begin{theorem}\label{theorem:venn}
The solutions of randomized Gonzalez's algorithm are a proper subset of the solutions of the randomized parametric pruning algorithm which are a proper subset of the $2$-approximation solutions for metric $k$-center. Also, the optimal solutions are not a subset of the solutions of the randomized parametric pruning algorithm.
\end{theorem}
The Venn diagram of sets in \Cref{theorem:venn} is shown in \Cref{fig:venn}, where the sets $2 \cdot OPT$ denote the set of $2$-approximation solutions and $OPT$ is the set of optimal solutions.
\begin{figure}[h]
\centering
\includegraphics[scale=0.8]{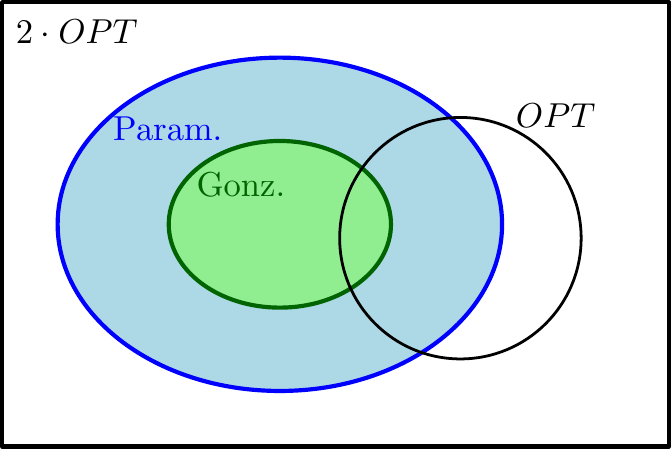}
\caption{The Venn diagram of worst-case solutions of randomized Gonzalez's algorithm (Gonz.), randomized parametric pruning algorithm (Param.), $2$-approximation ($2\cdot OPT$) and exact solutions ($OPT$) for $k$-ceneter.}
\label{fig:venn}
\end{figure}
\section{Efficient Parametric Pruning Algorithm}
We need to keep the time and space complexity of the coreset computation algorithm near linear. Since the time complexity of parametric pruning in general metrics is $O(n^3)$ and its space complexity is $O(n^2)$, we give a $(2+\epsilon)$-approximation algorithm with $O(\frac{nk}{\epsilon})$ time and $O(n)$ space.
Later, we use this algorithm to find a $(2+\epsilon)$-approximation algorithm for general metrics.

\begin{algorithm}[H]
\caption{Efficient Parametric Pruning}
\label{alg:memory}
\begin{algorithmic}[1]
\Require{A point set $P$, an integer $k$}
\Ensure{A set of centers $C$}
\State{$R=$ the radius of Gonzalez's $k$-center.}
\For{$r=R/2$, $k' < k$ and $r \leq 2(1+\epsilon)R$ ,$R\gets R(1+\epsilon)$}
\State{$k'=0$}
\State{$mark[1,\ldots,|P|]=false$}
\For{$i \in \{1,\ldots,|P|\}$}
\If{$mark[i]=false$}
\State{$mark[i]=true$}
\State{$C\gets C\cup \{P[i]\}$}
\State{$k'\gets k'+1$}
\For{$j \in \{1,\ldots,|P|\}$}
\If{$mark[j]=false$ and $dist(P[i],P[j])\leq R$}
\State{$mark[j]=true$}
\EndIf
\EndFor
\EndIf
\EndFor
\EndFor
\\ \Return{$C$}
\end{algorithmic}
\end{algorithm}

\begin{theorem}
The time complexity of \Cref{alg:memory} is $O(\frac{nk}{\epsilon})$ and its memory complexity is $O(n)$.
\end{theorem}
\begin{proof}
For each point, if it has been visited before, the algorithm ignores it, otherwise the algorithm chooses it as the next center, in which case it checks at most $n$ other points. Since there are at most $k$ centers, using aggregate method for amortized analysis, the running time of the algorithm is $nk+n$ for each $r$. The number of values $r$ that are checked by the algorithm is $1+\log_{1+\epsilon} 8$. Using taylor series $\ln (1+x) =x-\frac{x^2}{2}+...\geq x , x\rightarrow 0$, the overall time complexity is $O(\frac{nk}{\epsilon})$.
\end{proof}

\begin{theorem}\label{theorem:memory}
The approximation factor of \Cref{alg:memory} is $2+\epsilon$ for metric $k$-center.
\end{theorem}
\begin{proof}
Consider the disk graph of points with radius $2r^*$, where $r^*$ is the optimal radius. Using this radius, each cluster turns into a clique, so the maximal independent set subroutine of parametric pruning algorithm chooses at most one point from each cluster.

\Cref{alg:memory} computes a maximal independent set of the disk graph of radius $r$ at each step. For $r\geq 2r^*$, at most $k$ points are marked as centers by the algorithm. The algorithm starts from a lower bound on the radius and multiplies it by $(1+\epsilon)$. So, in the worst case the first radius that the algorithm checks which exceeds $2r^*$ is $r\leq 2(1+\epsilon)r^*$.
\end{proof}
\section{Connectivity Preservation and Applications to DBSCAN}
Computing the connected components of a graph is harder than testing the connectivity between two vertices $s,t$ of the graph.
It has been conjectured that sparse $st$-connectivity in $o(\log n)$ rounds and single-linkage clustering in high dimensions cannot be solved using a constant number of MapReduce rounds, by reduction from connectivity problem \cite{grigory}.

In DBSCAN, a point that has at least $f$ other points within distance $\epsilon$ from it is called a \textit{core point}. A cluster is a connected component of the intersection graph of balls of radius $\epsilon$ centered at core points. Any point that is not within distance $\epsilon$ of a core point is an \textit{outlier}.
Therefore, the algorithm can be seen as two main steps: simultaneous range counting queries, and computing the connected components of the disk graph. Both of these problems are challenging in the MapReduce model.

We use dual clustering to solve a non-convex clustering problem in MapReduce.
Several MapReduce algorithms for \textit{density-based spatial clustering of applications with noise (DBSCAN)} has been presented~\cite{mrdbscan,mrdbscan2,mrdbscan3,mrdbscan4,mrdbscan5}.
However, they lack theoretical guarantees.
We use the abstract DBSCAN algorithm \cite{dbscan2017}, which only differs from the original DBSCAN algorithm \cite{dbscan} in its time complexity \cite{dbscanhardness}, but computes the range counting queries prior to computing the connected components of the disk graph.

Several algorithms for range counting queries exist in MapReduce, but it is not possible in the $MRC$ model to run $n$ instances of single-query range search \cite{rangempc} simultaneously, since the data from one machine could be used in the solution of points from $n^{\delta}$ machines, for a constant $\delta>0$.
Mergeable summaries for range counting queries are randomized approximation algorithms which are also composable \cite{mergeable}.
Note that range queries for unit disks in $\mathbb{R}^d$ can be converted into rectangular range queries in $\mathbb{R}^{d+1}$, via linearization \cite{handbook}, therefore, any algorithm for rectangular range counting also solves the problem for disk range counting.

Our core-set for dual clustering of radius $\epsilon/2$, approximately preserves the connectivity of edges of weight at most $\epsilon$ between clusters.
\begin{lemma}\label{lemma:connectivity}
For two cluster centers $c_i,c_j$ of clusters $C_i,C_j$ of radius $\epsilon$, they are said to be connected if there is a point $p,d(p,c_i)\leq \epsilon, d(p,c_j)\leq \epsilon$. \Cref{alg:offline} with radius $\frac{\epsilon}{2}$ detects if such two cluster centers are connected or not.
\end{lemma}
\begin{proof}
By definition of clustering, the distances from each point to its cluster center is at most $\epsilon/2$, so
\[
d(p,c_i)\leq \epsilon/2, d(q,c_j)\leq \epsilon/2.
\]
If $d(p,q)\leq \epsilon$, then using triangle inequality twice gives the following results:
\[
d(c_i,c_j) \leq d(c_i,p)+d(p,c_j) \leq \epsilon/2+\epsilon/2 = \epsilon.
\]
\end{proof}

\Cref{alg:dbscan} with the minimum number of points of each cluster set to one, and then using the dual clustering, can be used to solve the DBSCAN problem in doubling metrics.

\begin{algorithm}[h]
\caption{A Coreset for DBSCAN of Core Points}
\label{alg:dbscan}
\begin{algorithmic}[1]
\Require{A set of sets of core points $\cup_{i=1}^L S_i$, a radius $\epsilon$}
\Ensure{A non-convex clustering of $\cup_{i=1}^L S_i$}
\State{$C_i$=The output of \Cref{alg:offline} with radius $\frac{\epsilon}{2}$ on set $S_i$.}
\State{Send centers from each machine to the first machine to form the set $C=\cup_i C_i$.}
\State{$T=$The connected components of the disk graph of radius $\epsilon$ on $C$.}
\State{Send $T$ to all machines.}
\State{Connect each point to its nearest neighbor in $T$.}
\end{algorithmic}
\end{algorithm}

\begin{theorem}
\Cref{alg:dbscan} solves DBSCAN using $O(1)$ rounds of MapReduce, given that $L=o(n), m =o(n),kL=O(m)$, where $k$ is the size of the output.
\end{theorem}
\begin{proof}
Using \Cref{lemma:connectivity}, $T$ has the same connected components for the points of $C$ as the optimal solution. Therefore, connecting each point to its nearest neighbor in $T$ gives an exact DBSCAN clustering.

Let $k$ be the number of points required to represent the clusters. Based on \Cref{theorem:dual}, the number of points returned by the algorithm is at most $D^2k=O(k)$.
Sending $O(k)$ data from $L$ machines to one machine requires $kL=O(m)$. 

Running \Cref{alg:offline} takes $O(1)$ rounds, computing the union and sending the clusters to all machines each take one round. So, the total number of rounds is $O(1)$.
\end{proof}
Note that in \Cref{alg:dbscan}, even without sending the points to a single machine, the set $C$ in \Cref{alg:dbscan} is still a composable core-set for DBSCAN.

\section{Experimental Results}
Description of data sets used in our experiments is depicted in \Cref{table:datasets}. Euclidean distance is used for all data sets. Note that DEXTER data set is not doubling, since it has a higher dimension than the number of its instances.
\begin{table}[h]
\centering
\begin{tabular}{|c|c|c|p{4cm}|}
\hline
Data Set & \# of Instances  & \# of Dimensions & Preprocessing\\
\hline
Parkinson \cite{parkinson} & 5875 & 26 & - \\
DEXTER \cite{dexter} & 2600 & 20000 & - \\
\hline
Power & 2049280 & 7 & No missing values, numerical attributes only\\
Higgs \cite{higgs} & 11000000 & 7 & $7$ high-level attributes only\\
\hline
\end{tabular}
\caption{Properties of data sets used in our experiments, obtained from \cite{datasets}.}
\label{table:datasets}
\end{table}
The size of data chunks used for partitioning the data is $m=10000$.
\subsection{Randomized Gonzalez vs. Randomized Parametric Pruning}
In \Cref{sec:compare}, we proved that the solutions of Gonzalez's algorithm are a subset of the solutions of the parametric pruning algorithm.
We compare the randomized version of these algorithms, where arbitrary choices in these algorithms are replaced by randomized ones. Then, we empirically compare the approximation factor of the resulting algorithms.

\begin{figure}[H]
\centering
\includegraphics[scale=0.6,trim={2cm 8.5cm 2cm 8.5cm},clip]{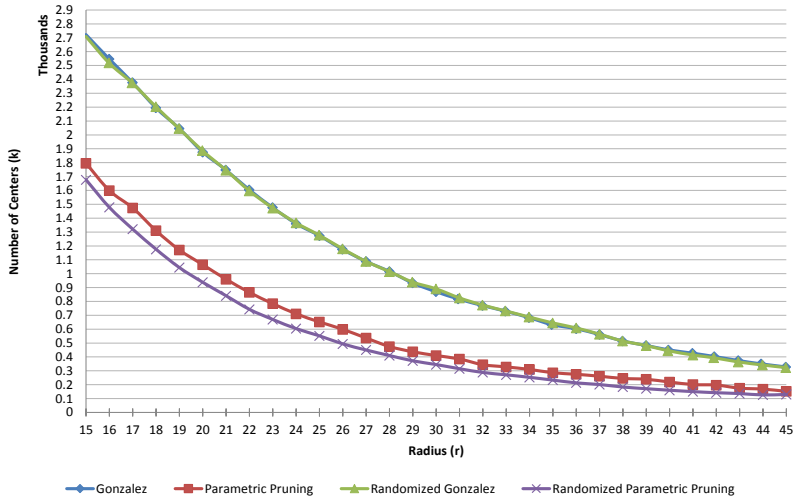}
\caption{A comparison of Gonzalez's greedy algorithm and the parametric pruning algorithm on Parkinson data-set.}
\label{fig:parkinson}
\end{figure}
\begin{figure}[H]
\centering
\includegraphics[scale=0.6,trim={2cm 8.8cm 2cm 8.8cm},clip]{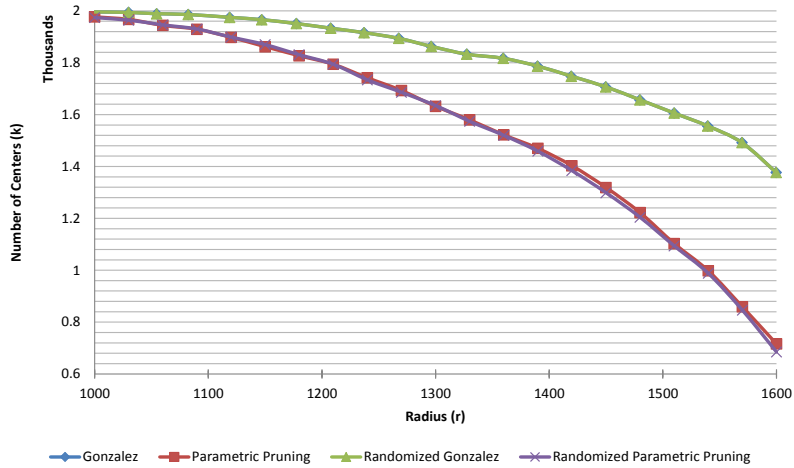}
\caption{A comparison of Gonzalez's greedy algorithm and the parametric pruning algorithm on Dexter data-set.}
\label{fig:dexter}
\end{figure}

The experiments show that the effect of randomization when choosing the points is slight, however, the differences between the approximation factor of the algorithms are more significant. In \Cref{fig:parkinson}, the results of the algorithm for a data-set in low dimensional Euclidean space, which is a doubling metric are given. \Cref{fig:dexter} shows the results for a high-dimensional Euclidean space, which is not necessarily a doubling metric.
\subsection{A Comparison in MapReduce}
In this experiment, we compared the approximation factor of the efficient parametric pruning algorithm (\Cref{alg:memory}) using $\epsilon=0.01$ with the greedy algorithm of Gonzalez extended to MapReduce \cite{gonzalez}.

\begin{figure}[H]
\centering
\includegraphics[scale=0.7,trim={2cm 9.5cm 2cm 10.9cm},clip]{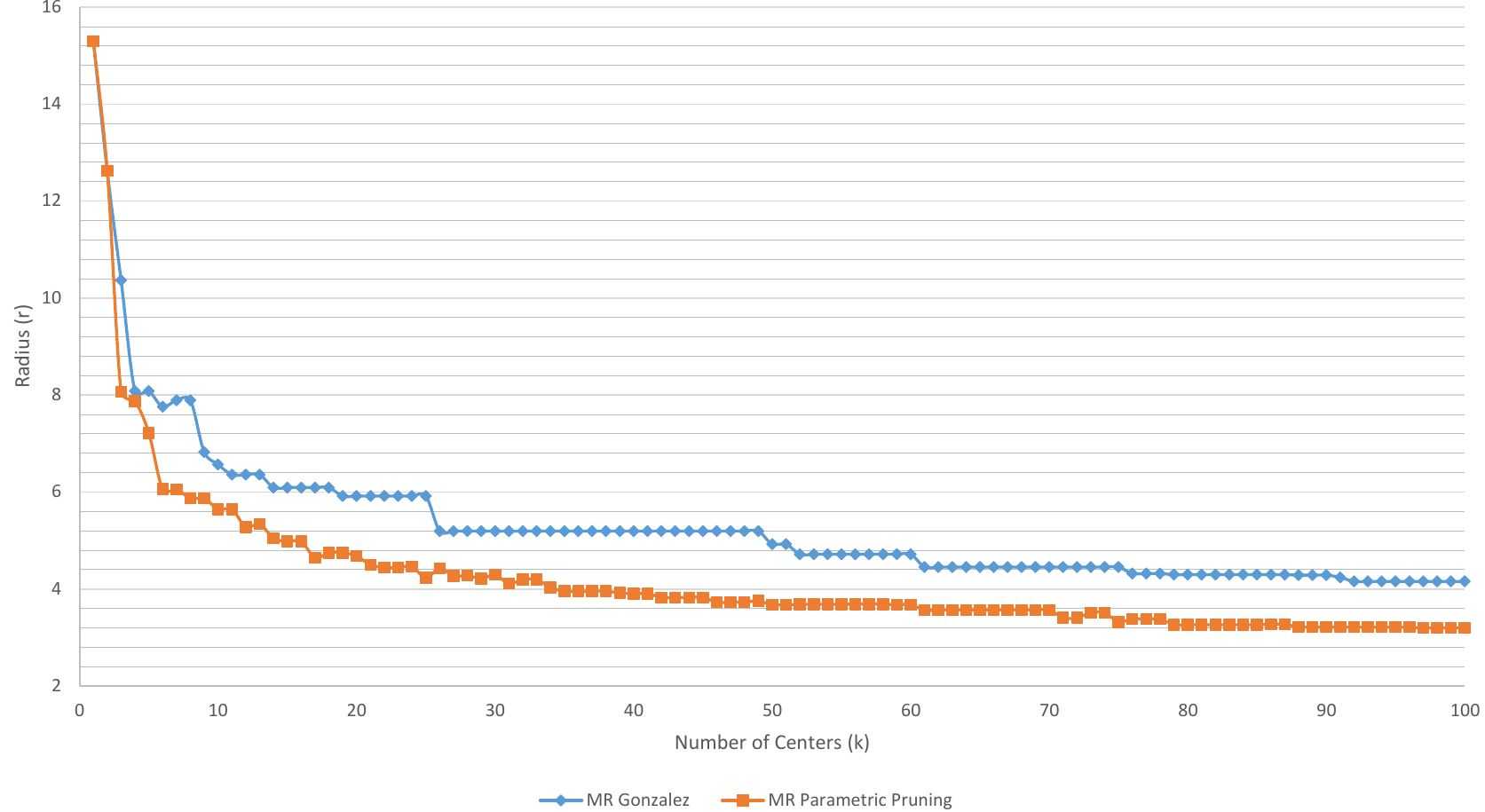}
\caption{A comparison of Gonzalez's greedy algorithm and the parametric pruning algorithm on Higgs data-set.}
\label{fig:higgs}
\end{figure}
\begin{figure}[H]
\centering
\includegraphics[scale=0.7,trim={2cm 9.5cm 2cm 11cm},clip]{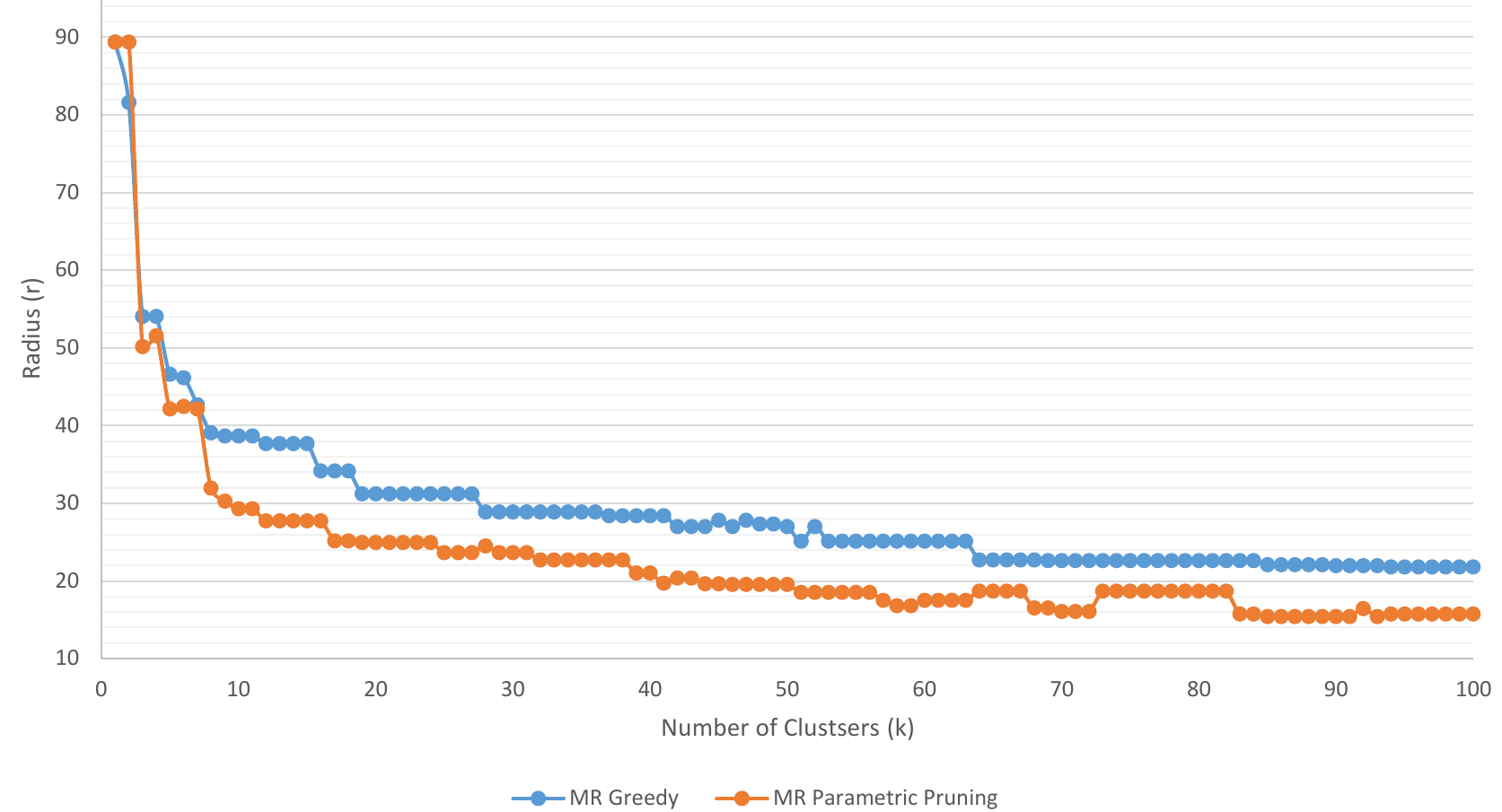}
\caption{A comparison of Gonzalez's greedy algorithm and the parametric pruning algorithm on Power data-set.}
\label{fig:power}
\end{figure}

The radii of Gonzalez's greedy algorithm in MapReduce are $1.3$ times the radii of parametric pruning algorithm on average on Higgs (\Cref{fig:higgs}) and Power (\Cref{fig:power}) data-sets.
\section{Conclusions}
Gonzalez's algorithm \cite{gonzalez} is a special case of parametric pruning algorithm \cite{vazirani} in which the greedy maximal independent set computation prioritizes the points with the maximum distance from the currently chosen points. 
Our algorithm and trade-off partially answers the open question of \cite{kcenter2} about comparing and improving these two algorithms in MapReduce. We propose a modified parametric pruning algorithm with running time $O(\frac{nk}{\epsilon})$ that achieves a better approximation factor in practice. Finding algorithms with provable approximation factor $2$ in the worst-case and better approximation factors on average remains open.

We also proved that the best possible trade-off between the approximation factor and the number of centers of $k$-center in doubling metrics is exponential.

Our composable coreset for dual clustering gives constant factor approximation for minimizing the size of DBSCAN cluster representatives given that the neighbor-counting is done prior to computing the coreset and the connected components. Finding a summarization technique that can preserve both the number of near neighbors and the connectivity between clusters in general metrics remains open. Note that in doubling metrics, keeping the number of points assigned to each center approximately solves this problem.

\bibliographystyle{abbrv}      
\bibliography{template.bbl}   

\end{document}